\documentclass[conference]{IEEEtran}
\IEEEoverridecommandlockouts
\usepackage[noadjust]{cite}
\usepackage{amsmath,amssymb,amsfonts}
\usepackage{algorithmic}
\usepackage{graphicx}
\usepackage{textcomp}
\usepackage{xcolor}
\usepackage{comment}
\usepackage{mathptmx}
\usepackage{amsthm}
\usepackage{dsfont}
\newtheorem{thm}{Theorem}
\newtheorem{lem}{Lemma}

\theoremstyle{definition}
\newtheorem{defn}{Definition}

\newtheorem{rem}{Remark}
\def\BibTeX{{\rm B\kern-.05em{\sc i\kern-.025em b}\kern-.08em
    T\kern-.1667em\lower.7ex\hbox{E}\kern-.125emX}}

\usepackage{bbm}

\newcommand{\indep}{\raisebox{0.05em}{\rotatebox[origin=c]{90}{$\models$}}}
\usepackage{etoolbox}    
\usepackage{dirtytalk}
\usepackage{changepage}
\usepackage{enumitem}
\newtoggle{singlecolumn}
\togglefalse{singlecolumn} 
\DeclareMathOperator*{\argmin}{\arg\!\min}

\newenvironment{sketchproofachievable}{%
   \proof}{\endproof}
\newenvironment{proofconverse}{%
  \proof}{\endproof}

\begin{document}

\title{Seeded Database Matching Under Noisy Column Repetitions\\
\thanks{This work is supported by National Science Foundation grants 1815821 and 2148293.}} 

\author{Serhat Bakirtas, Elza Erkip\\
 NYU Tandon School of Engineering\\
Emails: \{serhat.bakirtas, elza\}@nyu.edu }

\maketitle

\begin{abstract}
The re-identification or de-anonymization of users from anonymized data through matching with publicly-available correlated user data has raised privacy concerns, leading to the complementary measure of obfuscation in addition to anonymization. Recent research provides a fundamental understanding of the conditions under which privacy attacks are successful, either in the presence of obfuscation or synchronization errors stemming from the sampling of time-indexed databases. This paper presents a unified framework considering both obfuscation and synchronization errors and investigates the matching of databases under noisy column repetitions. By devising replica detection and seeded deletion detection algorithms, and using information-theoretic tools, sufficient conditions for successful matching are derived. It is shown that a seed size logarithmic in the row size is enough to guarantee the detection of all deleted columns. It is also proved that this sufficient condition is necessary, thus characterizing the database matching capacity of database matching under noisy column repetitions and providing insights on privacy-preserving publication of anonymized and obfuscated time-indexed data.
\end{abstract}

\section{Introduction}
\label{sec:introduction}
With the exponential boom in smart devices and the growing popularity of big data, companies and institutions have been gathering more and more personal data from users which is then either published or sold for research or commercial purposes. Although the published data is typically \emph{anonymized}, \emph{i.e.,} explicit identifiers of the users, such as names and dates of birth are removed, researchers~\cite{ohm2009broken} and companies~\cite{bigdata} have articulated their concerns over {the} insufficiency of anonymization for privacy as demonstrated by a series of practical attacks on real data~\cite{naini2015you,datta2012provable,narayanan2008robust,sweeney1997weaving,takbiri2018matching}. \emph{Obfuscation}, which refers to the deliberate addition of noise to the database entries, has been suggested as an additional measure to protect privacy~\cite{sweeney1997weaving}. While extremely valuable, this line of work does not provide a fundamental and rigorous understanding of the conditions under which anonymized and obfuscated databases are prone to privacy attacks. 

Recently, matching correlated pairs of databases have been investigated from an information-theoretic~\cite{cullina,shirani8849392,dai2019database,bakirtas2021database,noiselesslonger} and statistical~\cite{kunisky2022strong} points of view.
In~\cite{cullina}, Cullina \emph{et al.} proposed \emph{cycle mutual information} as a metric of correlation and derived sufficient and necessary conditions for successful matching, with the performance criterion being the perfect recovery for all users.
In~\cite{shirani8849392}, Shirani \emph{et al.} considered a pair of anonymized and obfuscated databases and drew analogies between database matching and channel decoding. By doing so, they derived sufficient and necessary conditions on the \emph{database growth rate} for reliable matching, in the presence of noise on the database entries. In~\cite{dai2019database} Dai \emph{et al.} investigated the matching of correlated databases with Gaussian attributes with the perfect recovery criterion. In~\cite{kunisky2022strong}, Kunisky and Niles-Weed investigated the same problem as Dai \emph{et al.}, from a statistical perspective, in different database size regimes for several performance criteria.
\begin{figure}[t]
\centerline{\includegraphics[width=0.5\textwidth,trim={0 15cm 2cm 0},clip]{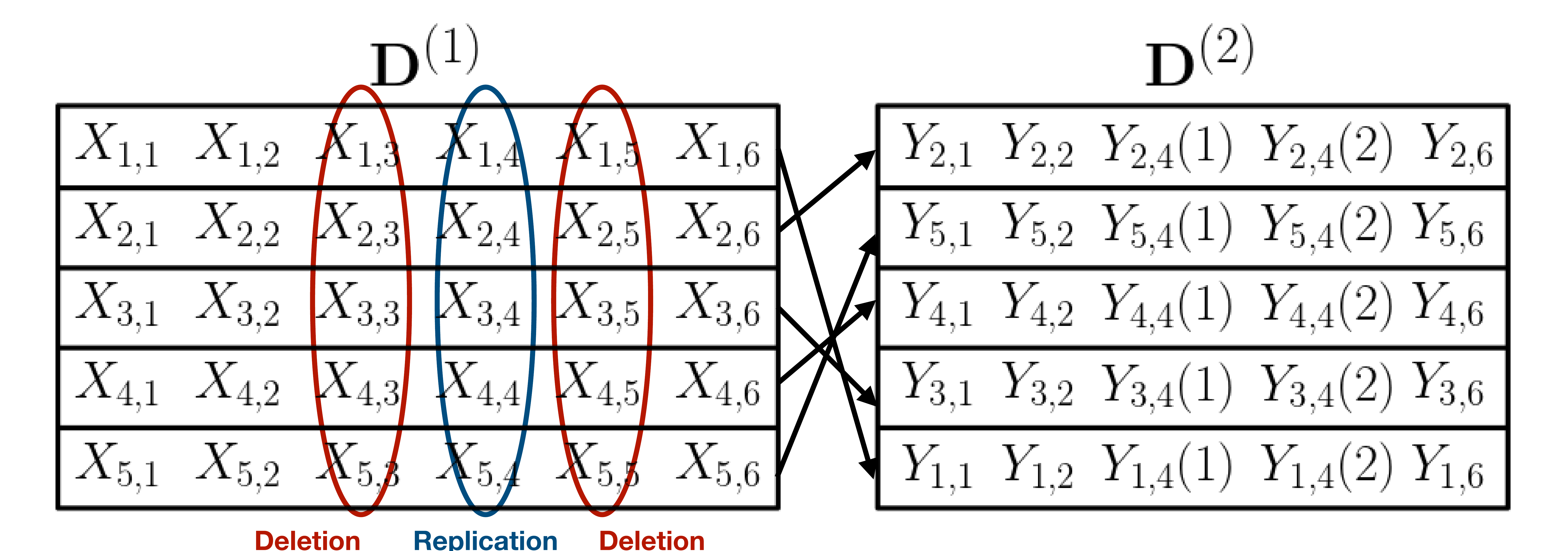}}
\caption{An illustrative example of database matching under noisy column repetitions. The columns circled in red are deleted whereas the fourth column, which is circled in blue, is repeated twice, \emph{i.e.,} replicated. For each $(i,j)$, $Y_{i,j}$ is the noisy observation of $X_{i,j}$. Furthermore, for each $i$, $Y_{i,4}(1)$ and $Y_{i,4}(2)$ are noisy replicas of $X_{i,4}$. Our goal is to estimate the row permutation {${\Theta}_n$} which is in this example given as; {$\boldsymbol{\Theta}_n(1)=5$}, {$\boldsymbol{\Theta}_n(2)=1$}, {$\boldsymbol{\Theta}_n(3)=4$},
{$\boldsymbol{\Theta}_n(4)=3$} and {$\boldsymbol{\Theta}_n(5)=2$}, by matching the rows of $\mathbf{D}^{(1)}$ and $\mathbf{D}^{(2)}$. Here the $i$\textsuperscript{th} row of $\mathbf{D}^{(1)}$ corresponds to the {$\Theta_n(i)$\textsuperscript{th}} row of $\mathbf{D}^{(2)}$.}
\label{fig:intro}
\end{figure}

In~\cite{bakirtas2021database}, motivated by the synchronization errors in the sampling of time-series datasets, we investigated the matching of two databases of the same number of users (rows), but with different numbers of attributes (columns). In our model, one of the databases suffers from \emph{random column deletions}, where the deletion indices are only partially and probabilistically available at the matching side. Under this side information assumption, we derived an achievable database growth rate. Demonstrating the impact of this side information on the achievable rate, we then proposed a \emph{deletion detection} algorithm given a batch of correctly-matched rows, \emph{i.e.,} \emph{seeds} and derived the seed size sufficient to guarantee a non-zero deletion detection probability.

In~\cite{noiselesslonger}, we investigated the matching of Markov databases, thus modeling correlations of the attributes (columns) under noiseless random column repetitions, a non-trivial extension of~\cite{bakirtas2021database}, where the attributes were assumed \emph{i.i.d.}. Under this generalized model, we devised a \emph{column histogram-based} repetition detection algorithm and derived an improved achievable rate, which is equal to the erasure bound~\cite{li2014input}. We then proved a converse showing the tightness of this achievable rate, thereby characterizing the exact matching capacity of Markov database matching under noiseless column repetitions.

In this paper, our goal is to investigate the necessary and the sufficient conditions for the successful matching of database rows under \emph{noisy} column repetitions. We assume a generalized database model where synchronization errors, in the form of column repetitions, are followed by noise, in the form of independent noise on the database entries{, as illustrated in Figure~\ref{fig:intro}}. The presence of noise prevents us from using the column histogram-based repetition detection algorithm of~\cite{noiselesslonger} and unlike~\cite{noiselesslonger} requires \emph{seed} users whose identities are known in both databases~\cite{bakirtas2021database,shirani2017seeded,fishkind2019seeded}. Under these assumptions, we devise two algorithms: one for deletion detection and the other for replica detection. We show that if the seed size $B$ grows linearly with the number of columns $n$, which is assumed to be logarithmic in the number of rows $m_n$ of the database, deletion locations can be extracted from the seeds. Then, we propose a joint typicality-based row matching scheme to derive sufficient {conditions} for successful matching. Finally, we prove a tight converse result, characterizing the matching capacity of the database matching problem under noisy column repetitions. 

The organization of this paper is as follows: Section~\ref{sec:problemformulation} contains the formulation of the problem. In Section~\ref{sec:achievability}, our main result on the matching capacity and its proof are presented. Finally, in Section~\ref{sec:conclusion} the results and ongoing work are discussed.

\noindent{\em Notation:} We denote the set of integers $\{1,...,n\}$ as $[n]$, and matrices with uppercase bold letters. For a matrix $\mathbf{D}$, $D_{i,j}$ denotes the $(i,j)$\textsuperscript{th} entry. Furthermore, by $A^n$, we denote a row vector consisting of scalars $A_1,\dots,A_n$ and the indicator of event $E$ by $\mathds{1}_E$. The logarithms, unless stated explicitly, are in base $2$. When the distinction is clear from the context, we use $\Theta$ to denote either the labeling function or the big theta notation for the asymptotic behavior.
\section{Problem Formulation}
\label{sec:problemformulation}

We use the following definitions, some of which are similar to~\cite{shirani8849392,bakirtas2021database,noiselesslonger}, to formally describe our problem. 

\begin{defn}{\textbf{(Unlabeled Database)}}\label{defn:unlabeleddb}
An ${(m_n,n,p_{X})}$ \emph{unlabeled database} is a randomly generated ${m_n\times n}$ matrix ${\mathbf{D}=\{D_{i,j}\in\mathfrak{X}\}}$ with \emph{i.i.d.} entries drawn according to the distribution $p_X$ with a finite discrete support $\mathfrak{X}=\{1,\dots,|\mathfrak{X}|\}$.
\end{defn}

\begin{defn}{\textbf{(Column Repetition Pattern)}}\label{defn:repetitionpattern}
The \emph{column repetition pattern} $S^n=\{S_1,S_2,...,S_n\}$ is a random vector consisting of $n$ \emph{i.i.d.} entries drawn from a discrete probability distribution $p_S$ with a finite integer support ${\{0,\dots,s_{\max}\}}$.
\end{defn}

\begin{figure}[t]
\centerline{\includegraphics[width=0.5\textwidth]{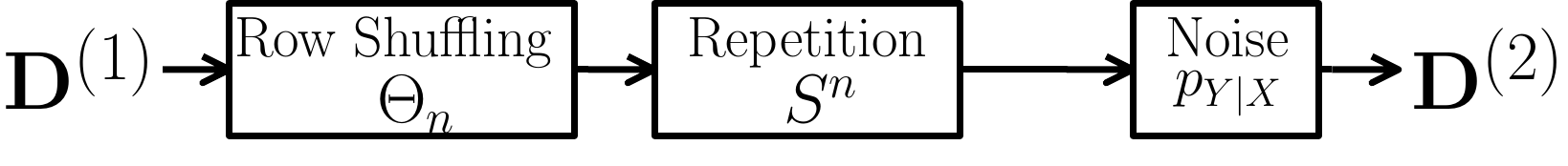}}
\caption{Relation between the unlabeled database $\mathbf{D}^{(1)}$ and the labeled noisy repeated one, $\mathbf{D}^{(2)}$.}
\label{fig:dmc}
\end{figure}

\begin{defn}{\textbf{(Labeled Noisy Repeated Database)}}\label{defn:labeleddb}
Let $\mathbf{D}^{(1)}$ be an ${(m_n,n,p_{X})}$ unlabeled database. Let $S^n$ be the independent repetition pattern, $\boldsymbol{\Theta}_n$ be a uniform permutation of $[m_n]$, independent of $(\mathbf{D}^{(1)},S^n)$ and $p_{Y|X}$ be a conditional probability distribution with both $X$ and $Y$ taking values from $\mathfrak{X}$. Given $\mathbf{D}^{(1)}$, {${S}^n$} and $p_{Y|X}$, $\mathbf{D}^{(2)}$ is called the \emph{labeled noisy repeated database} if the respective $(i,j)$\textsuperscript{th} entries ${D}^{(1)}_{i,j}$ and ${D}^{(2)}_{i,j}$ of $\mathbf{D}^{(1)}$ and $\mathbf{D}^{(2)}$ have the following relation:
\begin{align}
    D^{(2)}_{i,j}&=
    \begin{cases}
      E , &  \text{if } S_{j}=0\\
      Y_i^{S_j} & \text{if } S_{j}\ge 1
    \end{cases} \quad \forall i\in[m_n],\:\forall j\in[n]
\end{align}
where $Y_i^{S_j}$ is a random row vector of length $S_j$ with the following probability distribution, conditioned on $D^{(1)}_{\boldsymbol{\Theta}_n^{-1}(i),j}$
\vspace{-1em}
\begin{align}
    \Pr\left(Y_i^{S_j}=y^{S_j}\Big|{D}^{(1)}_{\boldsymbol{\Theta}_n^{-1}(i),j}\right)&=\prod\limits_{l=1}^{S_j} p_{Y|X}\left(y_l \Big|{D}^{(1)}_{\boldsymbol{\Theta}_n^{-1}(i),j}\right)\label{eq:noiseiid}
\end{align}
where $y^{S_j}=y_1,\dots,y_{S_j}$ and ${{D}^{(2)}_{i,j}=E}$ corresponds to ${D}^{(2)}_{i,j}$ being the empty string.

Note that $S_j$ indicates the times the $j$\textsuperscript{th} column of $\mathbf{D}^{(1)}$ is repeated. When $S_j=0$, the $j$\textsuperscript{th} column of $\mathbf{D}^{(1)}$ is said to be \emph{deleted} and when $S_j>1$, the $j$\textsuperscript{th} column of $\mathbf{D}^{(1)}$ is said to be \emph{replicated}.

The $i$\textsuperscript{th} row of $\mathbf{D}^{(2)}$ is said to correspond to the {$\boldsymbol{\Theta}_n^{-1}(i)$\textsuperscript{th}} row of $\mathbf{D}^{(1)}$, where $\boldsymbol{\Theta}_n$ is called the \emph{labeling function}.
\end{defn}

The relationship between $\mathbf{D}^{(1)}$ and $\mathbf{D}^{(2)}$, as described in Definiton~\ref{defn:labeleddb}, is illustrated in Figure~\ref{fig:dmc}.

Note that \eqref{eq:noiseiid} states that we can treat ${D}^{(2)}_{i,j}$ as the output of the discrete memoryless channel (DMC) $p_{Y|X}$ with input sequence consisting of $S_j$ copies of $\mathbf{D}^{(1)}_{\boldsymbol{\Theta}_n^{-1}(i),j}$ concatenated together. We stress that $p_{Y|X}$ is a general model, capturing any distortion and noise on the database entries, though we only refer to this as \say{noise} in this paper.

As we will discuss in Section~\ref{subsec:replicadetection}, in the noisy setting, inferring the column repetition pattern, particularly deletions, is a harder task compared to the noiseless setting investigated in~\cite{noiselesslonger}. Therefore, we assume the availability of \emph{seeds}, as done in noiseless database matching~\cite{bakirtas2021database} and graph matching \cite{shirani2017seeded, fishkind2019seeded} literatures.

\begin{defn}{\textbf{(Seeds)}}
For the unlabeled and labeled databases in Definitions~\ref{defn:unlabeleddb} and \ref{defn:labeleddb}, a \emph{seed} is a pair of correctly-matched rows. A \emph{batch of $B$ seeds} $(\mathbf{G}^{(1)},\mathbf{G}^{(2)})$ is a pair of databases (sub-matrices) with respective sizes $B\times n$ and $B\times \sum_{j=1}^n S_j$. We assume a polynomial \emph{seed size} $B=\Theta(n^d)$ where $d$ is called the \emph{seed order}.
\end{defn}

\begin{defn}{\textbf{(Successful Matching Scheme)}}
A \emph{matching scheme} is a sequence of mappings {${\phi_n: (\mathbf{D}^{(1)},\mathbf{D}^{(2)},\mathbf{G}^{(1)},\mathbf{G}^{(2)})\mapsto \hat{\boldsymbol{\Theta}}_n }$} where $\mathbf{D}^{(1)}$ is the unlabeled database, $\mathbf{D}^{(2)}$ is the labeled noisy repeated database, $(\mathbf{G}^{(1)},\mathbf{G}^{(2)})$ are seeds and $\hat{\boldsymbol{\Theta}}_n$ is the estimate of the correct labeling function $\boldsymbol{\Theta}_n$. The scheme $\phi_n$ is \emph{successful} if 
\begin{align}
    \Pr\left(\boldsymbol{\Theta}_n(J)\neq\hat{\boldsymbol{\Theta}}_n(J)\right)&\to 0  \text{ as }n\to\infty \label{eq:proberror}
\end{align}
where the index $J$ is drawn uniformly from $[m_n]$.
\end{defn}

Note that for a given column size $n$, as the row size $m_n$ increases, so does the probability of mismatch, as a result of having a larger number of candidates. Thus, in order to characterize the relationship between $m_n$ and $n$, we use the \emph{database growth rate} introduced in~\cite{shirani8849392}. As stated in~\cite[Theorem 1.2]{kunisky2022strong}, for distributions with parameters constant in $n$, the regime of interest is the logarithmic regime where $n\sim \log m_n$. 

\begin{defn}{\textbf{(Database Growth Rate)}}
The \emph{database growth rate} $R$ of an ${(m_n,n,p_X)}$ unlabeled database is defined as 
\begin{align}
    R&=\lim\limits_{n\to\infty} \frac{1}{n}\log m_n.
\end{align}
\end{defn}

\begin{defn}{\textbf{(Achievable Database Growth Rate)}}\label{defn:achievable}
Consider a sequence of ${(m_n,n,p_X)}$ unlabeled databases, a repetition probability distribution $p_S$, a noise distribution $p_{Y|X}$ and the resulting sequence of labeled noisy repeated databases. For a seed order $d$, a database growth rate $R$ is said to be \emph{achievable} if there exists a successful matching scheme when the unlabeled database has growth rate $R$.
\end{defn}

\begin{defn}{\textbf{(Matching Capacity)}}\label{defn:matchingcapacity}
The \emph{matching capacity} $C(d)$ is the supremum of the set of all achievable rates corresponding to a database distribution $p_X$, a repetition probability distribution $p_S$, a noise distribution $p_{Y|X}$ and a seed order $d$.
\end{defn}

In this paper, our goal is to characterize the matching capacity $C(d)$, by providing database matching schemes as well as a tight upper bound on all achievable database growth rates.

\section{Main Result}
\label{sec:achievability}
In this section, we present our main result on the matching capacity under noisy column repetitions (Theorem~\ref{thm:mainresult}) and prove its achievability by proposing a three-step approach: \emph{i)} noisy replica detection and \emph{ii)} deletion detection using seeds, followed by \emph{iii)} a row matching algorithm. Then, we outline the proof of the converse.

\begin{thm}{\textbf{(Matching Capacity Under Noisy Column Repetitions})}\label{thm:mainresult}
Consider a database distribution $p_X$, a column repetition distribution $p_S$ and a noise distribution $p_{Y|X}$. Then, for any seed order ${d\ge1}$, the matching capacity is
\begin{align}
    C(d) &= I(X;Y^S,S) \label{eq:matchingcap}
\end{align}
where $S\sim p_S$ and ${Y^S=Y_1,\dots,Y_S}$ such that
\begin{align}
    \Pr(Y^S=y_1,\dots,y_S|X=x)&=\prod\limits_{i=1}^S p_{Y|X}(y_i|x)
\end{align}
\end{thm}
Theorem~\ref{thm:mainresult} states that although the repetition pattern $S^n$ is not known a-priori, given a seed order $d\ge1$, we can achieve a database growth rate as if we knew $S^n$. Since the utility of seeds increase with the seed order $d$, we will focus on $d=1$, which we show is sufficient to achieve the matching capacity. {As we discuss in Section~\ref{subsec:converse}, the converse result holds for any seed size, whereas a general achievability result for the noisy case with $d<1$ requires additional combinatorial arguments and is omitted due to the space constraints.}

\begin{rem}{(\textbf{Noiseless Setting})}\label{cor:noiseless}
Using~\cite[Corollary 1]{noiselesslonger}, we can argue that in the noiseless setting, where
\begin{align}
    p_{Y|X}(y|x) &=\mathds{1}_{[y=x]}\:\forall x\in\mathfrak{X}
\end{align}
we have
\begin{align}
    C(d) = (1-\delta) H(X)
\end{align}
for any seed order $d$, where $\delta\triangleq p_S(0)$ is the deletion probability. Furthermore, we show in~\cite{noiselesslonger} that in the noiseless setting ${Y^S=X \otimes 1^S}$, the replicas do not offer any additional information. Thus, for any seed order $d\ge1$, Theorem~\ref{thm:mainresult} agrees with~\cite[Corollary 1]{noiselesslonger} in the noiseless setting with \emph{i.i.d.} columns. 
\end{rem}

\begin{rem}{(\textbf{No Synchronization Errors})}\label{cor:nosync}
As discussed in~\cite[Corollary 1]{shirani8849392}, when there are no synchronization errors, \emph{i.e.,} $p_S(1)=1$,
we have
\begin{align}
    C(d) = I(X;Y)
\end{align}
for any seed order $d$. Thus, under no synchronization errors, for any seed order $d\ge1$, Theorem~\ref{thm:mainresult} agrees with~\cite[Corollary 1]{shirani8849392}.
\end{rem}

The rest of this section is on the proof of Theorem~\ref{thm:mainresult}. In Section~\ref{subsec:replicadetection}, we discuss our noisy replica detection algorithm and prove its asymptotic performance. In Section~\ref{subsec:seededdeletiondetection}, we introduce a deletion detection algorithm which uses seeds and derive a seed size sufficient for an asymptotic performance guarantee. Then, in Section~\ref{subsec:matchingscheme}, we combine these two algorithms and prove the achievability of Theorem~\ref{thm:mainresult} by generalizing the rowwise matching scheme proposed in~\cite{noiselesslonger} to the noisy scenario. Finally, in Section~\ref{subsec:converse} we present the outline of the proof of the converse of Theorem~\ref{thm:mainresult}.

Note that when the two databases are independent, Theorem~\ref{thm:mainresult} states that the matching capacity becomes zero, hence our results trivially hold. Hence throughout this section, we assume that the two databases are not independent.

\subsection{Noisy Replica Detection}\label{subsec:replicadetection}
We propose to detect the replicas by extracting permutation-invariant features of the columns of $\mathbf{D}^{(2)}$. Our algorithm only considers the columns of $\mathbf{D}^{(2)}$ and as such, can only detect replications, not deletions. Furthermore, we stress that our replica detection algorithm does not require any seeds. 

In~\cite{noiselesslonger}, we chose the histogram of each column as its permutation-invariant feature, proved that the asymptotic uniqueness of the histograms and matched the column histograms of $\mathbf{D}^{(1)}$ and $\mathbf{D}^{(2)}$ to infer the repetition pattern. In the noisy setup, although still asymptotically-unique, the column histograms of the two databases cannot be matched due to noise. Joint typicality arguments do not work either, since arbitrary pairs of column histograms are likely to be jointly typical, even though the columns are independent. Therefore, we propose a replica detection algorithm which only considers $\mathbf{D}^{(2)}$ and adopts the \emph{Hamming distance between consecutive columns} of $\mathbf{D}^{(2)}$ as the permutation-invariant feature.

Let $K$ denote the number of columns of $\mathbf{D}^{(2)}$, $C^{m_n}_j$ denote the $j$\textsuperscript{th} column of $\mathbf{D}^{(2)}$, $j=1,\dots,K$. Our replica detection algorithm works as follows: We first compute the Hamming distances $d_H(C^{m_n}_j,C^{m_n}_{j+1})$ between $C^{m_n}_j$ and $C^{m_n}_{j+1}$, for {$j\in[K-1]$}.
For some average Hamming distance threshold $\tau$ chosen based on $p_{X,Y}$, the algorithm decides that $C^{m_n}_{j}$ and $C^{m_n}_{j+1}$ are replicas only if $d_H(C^{m_n}_{j},C^{m_n}_{j+1})<m_n \tau$, and independent otherwise. In the following lemma, we show that this algorithm can infer the replicas with high probability. 

\begin{lem}{\textbf{(Noisy Replica Detection)}}\label{lem:noisyreplicadetection}
Let $E_j$ denote the event that the Hamming distance based algorithm described above fails to infer the correct relationship between the columns $C^{m_n}_{j}$ and $C^{m_n}_{j+1}$ of $\mathbf{D}^{(2)}$, $j=1,\dots,K-1$. Then
\begin{align}
    \Pr(\bigcup\limits_{j=1}^{K-1} E_j)&\to 0\text{ as }n\to\infty
\end{align}
\end{lem}
\begin{proof}
Let $(X_1,Y_1),(X_2,Y_2)\sim p_{X,Y}$ be two pairs of random variables. We define
\begin{align}
    p_0&\triangleq \Pr(Y_1\neq Y_2|X_1 \indep X_2)\\
    p_1&\triangleq \Pr(Y_1\neq Y_2|X_1=X_2)
\end{align}
Observe that $Y_1$ and $Y_2$ are noisy observations of independent database entries $X_1,\:X_2$ when $X_1 \indep X_2$ and $Y_1$ and $Y_2$ are noisy replicas when $X_1=X_2$.
We can rewrite $p_0$ and $p_1$ as the following.
 \begin{align}
    p_0 &=\sum\limits_{x_1\in\mathfrak{X}}\sum\limits_{x_2\in\mathfrak{X}}\sum\limits_{y\in\mathfrak{X}} p_X(x_1) p_X(x_2) p_{Y|X}(y|x_1) \left[1-p_{Y|X}(y|x_2)\right]\\
    &=\sum\limits_{x_1\in\mathfrak{X}}\sum\limits_{y\in\mathfrak{X}} p_X(x_1)  p_{Y|X}(y|x_1)\sum\limits_{x_2\in\mathfrak{X}}p_X(x_2) \left[1-p_{Y|X}(y|x_2)\right]\\
    &=\sum\limits_{x\in\mathfrak{X}}\sum\limits_{y\in\mathfrak{X}} p_X(x)  p_{Y|X}(y|x)\left[1-p_Y(y)\right]\\
    p_1 &=\sum\limits_{x\in\mathfrak{X}}\sum\limits_{y\in\mathfrak{X}} p_X(x) p_{Y|X}(y|x) \left[1-p_{Y|X}(y|x)\right]
\end{align}
Thus, we have 
\begin{align}
    p_0 - p_1 &=\sum\limits_{x\in\mathfrak{X}}\sum\limits_{y\in\mathfrak{X}} p_{X,Y}(x,y)\left[p_{Y|X}(y|x)-p_Y(y)\right]
\end{align}
For every $y\in\mathfrak{X}$, let
\begin{align}
    \psi(y) &\triangleq \sum\limits_{x\in\mathfrak{X}}p_X(x)\left[p_{Y|X}(y|x)-p_Y(y)\right]^2\\
    &=\sum\limits_{x\in\mathfrak{X}}p_X(x)\left[p_{Y|X}(y|x)-\sum\limits_{z\in\mathfrak{X}}p_{Y|X}(y|z)p_X(z)\right]^2\\
    &\ge 0 \label{eq:psinonnegative}
\end{align}
where \eqref{eq:psinonnegative} follows from the non-negativity of the square term in the summation. It must be noted that $\psi(y)=0$ only if $p_{Y|X}(y|x)=p_Y(y) \forall x\in\mathfrak{X}$ with $p_X(x)>0$.

Now, expanding the square term, we obtain
\begin{align}
    \psi(y)&= \sum\limits_{x\in\mathfrak{X}}p_X(x) p_{Y|X}(y|x)^2-2 p_Y(y) \sum\limits_{x\in\mathfrak{X}}p_X(x) p_{Y|X}(y|x)\notag\\
    &\qquad+\sum\limits_{x\in\mathfrak{X}}p_X(x) p_{Y}(y)^2\\
    &= \sum\limits_{x\in\mathfrak{X}}p_X(x) p_{Y|X}(y|x)^2 - 2 p_Y(y)^2 +p_Y(y)^2\\
    &= \sum\limits_{x\in\mathfrak{X}}p_X(x) p_{Y|X}(y|x)^2 - p_Y(y)^2
\end{align}
Now, we rewrite $p_0-p_1$ as 
\begin{align}
    p_0-p_1&=\sum\limits_{y\in\mathfrak{X}}\sum\limits_{x\in\mathfrak{X}} p_{X,Y}(x,y)\left[p_{Y|X}(y|x)-p_Y(y)\right]\\
    &= \sum\limits_{y\in\mathfrak{X}}\left[\left(\sum\limits_{x\in\mathfrak{X}} p_X(x) p_{Y|X}(y|x)^2\right)-p_Y(y)^2\right]\\
    &= \sum\limits_{y\in\mathfrak{X}} \psi(y)\\
    &\ge 0
\end{align}
with $p_0-p_1=0$ only when $p_{Y|X}(y|x)=p_Y(y)\forall x,y\in\mathfrak{X}$. In other words, $p_0>p_1$ as long as the two databases are not independent.

Choose any $\tau\in(p_1,p_0)$ bounded away from both $p_0$ and $p_1$. Let $A_j$ denote the event that $C^{m_n}_{j}$ and $C^{m_n}_{j+1}$ are replicas and $B_j$ denote the event that the algorithm detects $C^{m_n}_{j}$ and $C^{m_n}_{j+1}$ as replicas. From the union bound,
\begin{align}
    \Pr(\bigcup\limits_{j=1}^{K-1} E_j)&\le \sum\limits_{j=1}^{K-1} \Pr(A_j ^c) \Pr(B_j|A_j ^c)+ \Pr(A_j)  \Pr(B_j^c|A_j)\label{eq:replicadetectionbound}
\end{align}
Note that conditioned on $A_j^c$, $d_H(C^{m_n}_j,C^{m_n}_{j+1})\sim\text{Binom}(m_n,p_0)$ and conditioned on $A_j$, $d_H(C^{m_n}_j,C^{m_n}_{j+1})\sim\text{Binom}(m_n,p_1)$. Then, from Chernoff bound~\cite[Theorem 1]{hoeffding1994probability}, we get
\begin{align}
    \Pr(B_j|A_j ^c)&\le {2}^{-m_n D\left(\tau\|p_0\right)}\label{eq:chernoff1}\\
    \Pr(B_j^c|A_j)&\le {2}^{-m_n D\left((1-\tau)\|1-p_1\right)}\label{eq:chernoff2}
\end{align}
{where $D(.\|.)$ denotes the Kullback-Leibler divergence~\cite[Chapter 2.3]{cover2006elements} between two Bernoulli distributions with given parameters.}
Thus, we get
\begin{align}
    \Pr(\bigcup\limits_{j=1}^{K-1} E_j)&\le (K-1)\left[ {2}^{-m_n D\left(\tau\|p_0\right)}+( {2}^{-m_n D\left((1-\tau)\|1-p_1\right)}\right]
\end{align}
Observing that RHS of \eqref{eq:replicadetectionbound} has $2K-2=O(n)$ terms decaying exponentially in~$m_n$ and $n\sim\log m_n$ concludes the proof. 
\end{proof}

\subsection{Deletion Detection Using Seeds}\label{subsec:seededdeletiondetection}
e propose to detect deletions using seeds. Let $(\mathbf{G}^{(1)},\mathbf{G}^{(2)})$ be a batch of $B=\Theta(n^d)$ seeds. Our deletion detection algorithm works as follows: After finding the replicas as {in} Section~\ref{subsec:replicadetection}, we discard all-but-one of the noisy replicas from $\mathbf{G}^{(2)}$, to obtain $\tilde{\mathbf{G}}^{(2)}$ whose column size is denoted by $\tilde{K}$. At this step, we only have deletions. 

We adopt an exhaustive search over all potential deletion patterns with $n-\tilde{K}$ deletions on $\mathbf{G}^{(1)}$. For each deletion pattern $I$, we compute the total Hamming distance $d_H(\tilde{\mathbf{G}}^{(1)}(I),\tilde{\mathbf{G}}^{(2)})$ between $\tilde{\mathbf{G}}^{(1)}(I)$ and $\tilde{\mathbf{G}}^{(2)}$, where ${\tilde{\mathbf{G}}^{(1)}(I)}$ denotes the matrix obtained by discarding the columns whose indices lie in $I$ from $\mathbf{G}^{(1)}$. More formally, we compute
\begin{align}
    d_H(\tilde{\mathbf{G}}^{(1)}(I),\tilde{\mathbf{G}}^{(2)}) &= \sum_{i\in[m_n]}\sum_{j\in[n-\tilde{K}]} {\mathds{1}_{\left[\tilde{{G}}^{(1)}(I)_{i,j}\neq \tilde{{G}}^{(2)}_{i,j}\right]}}
\end{align}
Then, the algorithm outputs the deletion pattern minimizing total Hamming distance between $\tilde{\mathbf{G}}^{(1)}(I)$ and $\tilde{\mathbf{G}}^{(2)}$, denoted by $\hat{I}_{\text{del}}$. In other words, 
\begin{align}
    \hat{I}_{\text{del}} = \argmin\limits_{I\subseteq [n], |I|=n-\tilde{K}} d_H(\tilde{\mathbf{G}}^{(1)}(I),\tilde{\mathbf{G}}^{(2)})
\end{align}

Note that such a strategy depends on pairs of correlated entries in $\mathbf{G}^{(1)}$ and $\tilde{\mathbf{G}}^{(2)}$ having a higher probability of being equal than independent pairs. More formally, given a correlated pair ${(X_1,Y_1)}\sim p_{X,Y}$, and an independent pair ${(X_2,Y_1)}\sim p_{X}p_{Y}$ we need
\begin{align}
    \Pr(Y_{1}= X_{1})>\Pr(Y_{1}= X_{2})\label{eq:conditiondeletiondetection}
\end{align}
which is not true in general. 

For example, suppose ${\mathfrak{X}=\{0,1\}}$ with ${p_X(0)=1/2}$ and $p_{Y|X}$ follows BSC($q$), \emph{i.e.} ${p_{Y|X}(x|x)=1-q}$, ${x=0,1}$. Note that when ${q>1/2}$ \eqref{eq:conditiondeletiondetection} is not satisfied. However, we can flip the output bits, by applying the bijective remapping ${\sigma=\left(\begin{smallmatrix}
1 & 2\\
2 & 1
\end{smallmatrix}\right)}$ to $Y$ in order to satisfy~\eqref{eq:conditiondeletiondetection}. 

Thus, as long as such a bijective remapping ${\sigma:\mathfrak{X}\to\mathfrak{X}}$ satisfying \eqref{eq:conditiondeletiondetection} exists, we can use the aforementioned deletion detection algorithm.
Now, suppose that such a mapping $\sigma$ exists. We apply $\sigma$ to the entries of $\tilde{\mathbf{G}}^{(2)}$ to construct $\tilde{\mathbf{G}}_{\sigma}^{(2)}$. Then, our deletion detection algorithm computes $d_H(\tilde{\mathbf{G}}^{(1)}(I),\tilde{\mathbf{G}}_{\sigma}^{(2)})$ for each potential deletion pattern $I$ and outputs the pattern $\hat{I}_{\text{del}}(\sigma)$ minimizing it. In other words, 
\begin{align}
    \hat{I}_{\text{del}}(\sigma) = \argmin\limits_{I\subseteq [n], |I|=n-\tilde{K}} d_H(\tilde{\mathbf{G}}^{(1)}(I),\tilde{\mathbf{G}}_{\sigma}^{(2)})
\end{align}

The following lemma states that such a bijective mapping $\sigma$ exists and for a seed order $d \ge 1$, this algorithm can infer the deletion locations with high probability.

\begin{lem}{\textbf{(Seeded Deletion Detection)}}\label{lem:seededdeletiondetection}
 For a repetition pattern ${S}^n$, let ${I_\text{del}=\{j\in[n]|S_j=0\}}$. Then there exists a bijective mapping $\sigma$ depending on $p_{X,Y}$ satisfying \eqref{eq:conditiondeletiondetection} and for seed order $d=1$, 
 \vspace{-0.5em}
\begin{align}
    \Pr\left(\hat{I}_\text{del}(\sigma)=I_\text{del}\right)&\to 1\text{ as }n\to\infty
\end{align}
\end{lem}
\begin{proof}
We first prove the existence of such a bijective mapping $\sigma$, satisfying \eqref{eq:conditiondeletiondetection}. For all $\sigma$, let
\begin{align}
    q_0(\sigma)&=\Pr(\sigma(Y_1)\neq X_2)\notag\\
    &\triangleq \sum\limits_{x_1\in\mathfrak{X}}\sum\limits_{x_2\in\mathfrak{X}}p_X(x_1) p_X(x_2)[1-p_{Y|X}(\sigma^{-1}(x_2)|x_1)]\\
    q_1(\sigma)&\triangleq\Pr(\sigma(Y_1)\neq X_1)\notag\\
    &= \sum\limits_{x\in\mathfrak{X}}p_X(x) [1-p_{Y|X}(\sigma^{-1}(x)|x)]
\end{align}
Here, our goal is to show that there exists at least one $\sigma$ satisfying
\begin{align}
    q_0(\sigma)>q_1(\sigma)\label{eq:q0biggerthanq1}
\end{align}

We first prove 
\begin{align}
    \sum\limits_{\sigma} q_0(\sigma)-q_1(\sigma)=0\label{eq:q0q1overphizero}
\end{align}
where the summation is over all permutations $\sigma$. For brevity, let 
\begin{align}
    P_{i,j}\triangleq p_{Y|X}(j|i)\quad \forall i,j\in\mathfrak{X} \label{eq:definep}
\end{align}

Note that from \eqref{eq:definep}, we have
\begin{align}
    \sum\limits_{j=1}^{|\mathfrak{X}|}& P_{i,j}=1\quad \forall i\in \mathfrak{X}\label{eq:pijsumto1}\\
    \sum\limits_{i=1}^{|\mathfrak{X}|}&\sum\limits_{j=1}^{|\mathfrak{X}|} P_{i,j}=|\mathfrak{X}|\label{eq:pijsumtoX}
\end{align}
Taking the sum over all $\sigma$, we obtain
\begin{align}
    \sum\limits_{\sigma} q_0(\sigma)-{q_1(\sigma)}
    &= \sum\limits_{\sigma}\sum\limits_{i=1}^{|\mathfrak{X}|}\sum\limits_{j=1}^{|\mathfrak{X}|} p_X(i) p_X(j) P_{i,\sigma^{-1}(j)}\notag \\ &\hspace{4em}-\sum\limits_{\sigma}\sum\limits_{i=1}^{|\mathfrak{X}|} p_X(i) P_{i,\sigma^{-1}(i)}\label{eq:q0q1sumoverphi}
\end{align}
Combining \eqref{eq:pijsumto1}-\eqref{eq:q0q1sumoverphi}, it can be shown that both terms on the RHS of \eqref{eq:q0q1sumoverphi} are equal to $(|\mathfrak{X}|-1)!$. Thus, we have proved~\eqref{eq:q0q1overphizero}. 

Now, we only need to show that
\begin{align}
    \exists \sigma\quad q_0(\sigma)-q_1(\sigma)\neq 0
\end{align}
Considering several one-cycle permutations over $\mathfrak{X}$, one can show that
\begin{align}
    q_0(\sigma)-q_1(\sigma)= 0\hspace{0.5em} \forall \sigma \iff p_{Y|X}(y|x)=p_Y(y)\hspace{0.5em} \forall (x,y)\in\mathfrak{X}^2 \label{eqn:phicondition}
\end{align}
We have assumed the databases are not independent, \emph{i.e.,} $p_{X,Y}\neq p_X p_Y$. Thus, there exists a bijective mapping $\sigma$ satisfying~\eqref{eq:q0biggerthanq1}.

Now choose such a mapping $\sigma$. Let ${\hat{K}=\sum_{j=1}^n \mathbbm{1}_{[S_j\neq 0]}}$ and $\Lambda_n$ be the seed size. Let $\epsilon>0$ and declare error if $\hat{K}\notin[(1-\delta-\epsilon)n,(1-\delta+\epsilon)n]$ whose probability is denoted by $\kappa_n$. Then, we use the union bound to obtain
\vspace{-1em}
\begin{adjustwidth}{-0.3cm}{0pt}
\begin{align}
    \Pr&\left(\hat{I}_\text{del}(\sigma)\neq I_\text{del}\right)\le \kappa_n+ \notag\\
    &\sum\limits_{I\subseteq [n], |I|=\hat{K}}\Pr(d_H(\tilde{\mathbf{G}}^{(1)}(I),\tilde{\mathbf{G}}_{\sigma}^{(2)})\le d_H(\tilde{\mathbf{G}}^{(1)}(I_{\text{del}}),\tilde{\mathbf{G}}_{\sigma}^{(2)}) )\label{eq:deletiondetection}
\end{align}
\end{adjustwidth}
where the difference of the total Hamming distances in \eqref{eq:deletiondetection} can be written as the difference of two Binomial random variables with a common number of trials depending on the size of the overlap between $I$ and $I_{\text{del}}$. 

Specifically, denote by $f(I,I_{\text{del}})$ the number of overlapping elements between $[n]\setminus I$ and $[n]\setminus I_{\text{del}}$. Here we count the overlaps as follows: We count $i_1\in ([n]\setminus I)\bigcap [n]\setminus I_{\text{del}}$ as an overlapping element only if $i_1$ is in the same position in each one of the ordered sets $i_1\in ([n]\setminus I)$ and $[n]\setminus I_{\text{del}}$. For example, let $n=3$, $I=\{1\}$, $I_{\text{del}}=\{3\}$. Then we have $[n]\setminus I=\{2,3\}$ and $[n]\setminus I_{\text{del}}=\{1,2\}$. Note that even though the element $2$ is present in both sets, it is in different positions when the sets are ordered. In this case, we have $f(I,I_{\text{del}})=0$. 

Now, observe that 
\begin{align}
    d_H(\tilde{\mathbf{G}}^{(1)}(I),\tilde{\mathbf{G}}_{\sigma}^{(2)})- d_H(\tilde{\mathbf{G}}^{(1)}(I_{\text{del}}),\tilde{\mathbf{G}}_{\sigma}^{(2)}) 
\end{align}
can be written as the difference between two Binomial random variables with respective parameters ${(\Lambda_n(\hat{K}-f(I,I_{\text{del}})),q_0(\sigma))}$ and ${(\Lambda_n(\hat{K}-f(I,I_{\text{del}})),q_1(\sigma))}$.
From Hoeffding's inequality~\cite{hoeffding1994probability}, we obtain 
\begin{align}
    \Pr(&d_H(\tilde{\mathbf{G}}^{(1)}(I),\tilde{\mathbf{G}}_{\sigma}^{(2)})\le d_H(\tilde{\mathbf{G}}^{(1)}(I_{\text{del}}),\tilde{\mathbf{G}}_{\sigma}^{(2)}) )\notag\\
    &= \Pr(d_H(\tilde{\mathbf{G}}^{(1)}(I),\tilde{\mathbf{G}}_{\sigma}^{(2)})- d_H(\tilde{\mathbf{G}}^{(1)}(I_{\text{del}}),\tilde{\mathbf{G}}_{\sigma}^{(2)})\le 0 )\\
    &\le \exp\left(-\frac{1}{2} \Lambda_n(\hat{K}-f(I,I_{\text{del}})) (q_0(\sigma)-q_1(\sigma))^2 \right)\\
    &= q^{\Lambda_n(\hat{K}-f(I,I_{\text{del}}))}
\end{align}
where 
\begin{align}
    q\triangleq e^{-\frac{1}{2}(q_0(\sigma)-q_1(\sigma))^2}<1
\end{align}
Furthermore, the number of false deletion index sets $I$ with a given $f(I,I_{\text{del}})$ can be wastefully upper bounded by $\binom{n}{\hat{K}}$. Thus, we can further bound the probability of error as
\begin{align}
    \Pr\left(\hat{I}_\text{del}(\sigma)\neq I_\text{del}\right)&\le \kappa_n + \sum\limits_{i=0}^{\hat{K}-1} \binom{n}{\hat{K}} q^{\Lambda_n(\hat{K}-i)}\\
    &= \kappa_n + \binom{n}{\hat{K}} \sum\limits_{i=0}^{\hat{K}-1} q^{\Lambda_n(\hat{K}-i)}\\
    &= \kappa_n + \binom{n}{\hat{K}} \sum\limits_{j=1}^{\hat{K}} q^{\Lambda_n j}\\
    &= \kappa_n + \binom{n}{\hat{K}} \sum\limits_{i=0}^{\hat{K}-1} q^{\Lambda_n (i+1)}\\
    &= \kappa_n + \binom{n}{\hat{K}} \sum\limits_{i=0}^{\hat{K}-1} q^{\Lambda_n} q^{\Lambda_n i}\\
    &= \kappa_n +\binom{n}{\hat{K}} q^{\Lambda_n} \sum\limits_{i=0}^{\hat{K}-1}  q^{\Lambda_n i}\\
    &\le \kappa_n +2^{n H_b(\hat{K}/n)} q^{\Lambda_n} \frac{1-q^{\Lambda_n  \hat{K}}}{1-q^{\Lambda_n}}\\
    &\le \kappa_n +2^{n H_b(\hat{K}/n)} q^{\Lambda_n} \frac{1}{1-q}\\
    &= \kappa_n +\frac{1}{1-q} 2^{n H_b(\hat{K}/n)-\Lambda_n\log\frac{1}{q}}\label{eq:deldetfinal}
\end{align}
where $H_b$ denotes the binary entropy function. Observe that the RHS of \eqref{eq:deldetfinal} vanishes as $n\to\infty$ if 
\begin{align}
    \Lambda_n\ge \frac{n H_b(\hat{K}/n)}{\log\frac{1}{q}}{=\frac{2 n H_b(\hat{K}/n)}{(q_0(\sigma)-q_1(\sigma))^2 \log e}}
\end{align}
which can be satisfied with some $\Lambda_n=\Theta(n)$.
Thus a seed order $d=1$ is sufficient for successful deletion detection.
\end{proof}
In contrast with the linear seed size of Lemma~\ref{lem:seededdeletiondetection}, \cite{bakirtas2021database} requires that the number of seeds is logarithmic in the number of columns. This is because in~\cite{bakirtas2021database} the performance criterion is the successful detection of an \emph{arbitrarily-chosen} deleted column, whereas in this work, the criterion is the successful detection of \emph{all} deleted columns.
\subsection{Row Matching Scheme and Achievability}\label{subsec:matchingscheme}
We are now ready to outline the proof of achievability of Theorem~\ref{thm:mainresult}.
\vspace*{-0.75em}
\begin{sketchproofachievable}
Let $S^n$ be the underlying repetition pattern and $K\triangleq\sum_{i=1}^n S_i$ be the number of columns in $\mathbf{D}^{(2)}$. The matching scheme we propose follows these steps:
\begin{enumerate}[label=\textbf{ \arabic*)},leftmargin=1.3\parindent]
\item Perform replica detection as in Section~\ref{subsec:replicadetection}. The probability of error of this step is denoted by $\rho_n$.
\item Perform deletion detection using seeds as in Section~\ref{subsec:seededdeletiondetection}. The probability of error is denoted by $\mu_n$. At this step, we have an estimate $\hat{S}^n$ of $S^n$.
\item Using $\hat{S}^n$, place markers between the noisy replica runs of different columns to obtain $\tilde{\mathbf{D}}^{(2)}$. If a run has length 0, \emph{i.e.} deleted, introduce a column consisting of erasure symbol $\ast\notin\mathfrak{X}$. Note that provided that the detection algorithms in Steps~1 and 2 have performed correctly, there are exactly $n$ such runs, where the $j$\textsuperscript{th} run in $\tilde{\mathbf{D}}^{(2)}$ corresponds to the noisy copies of the $j$\textsuperscript{th} column of {$\Theta_n\circ{\mathbf{D}}^{(1)}$} if $S_j\neq 0$, and an erasure column otherwise. 
\begin{figure}[t]
\centerline{\includegraphics[width=0.48\textwidth]{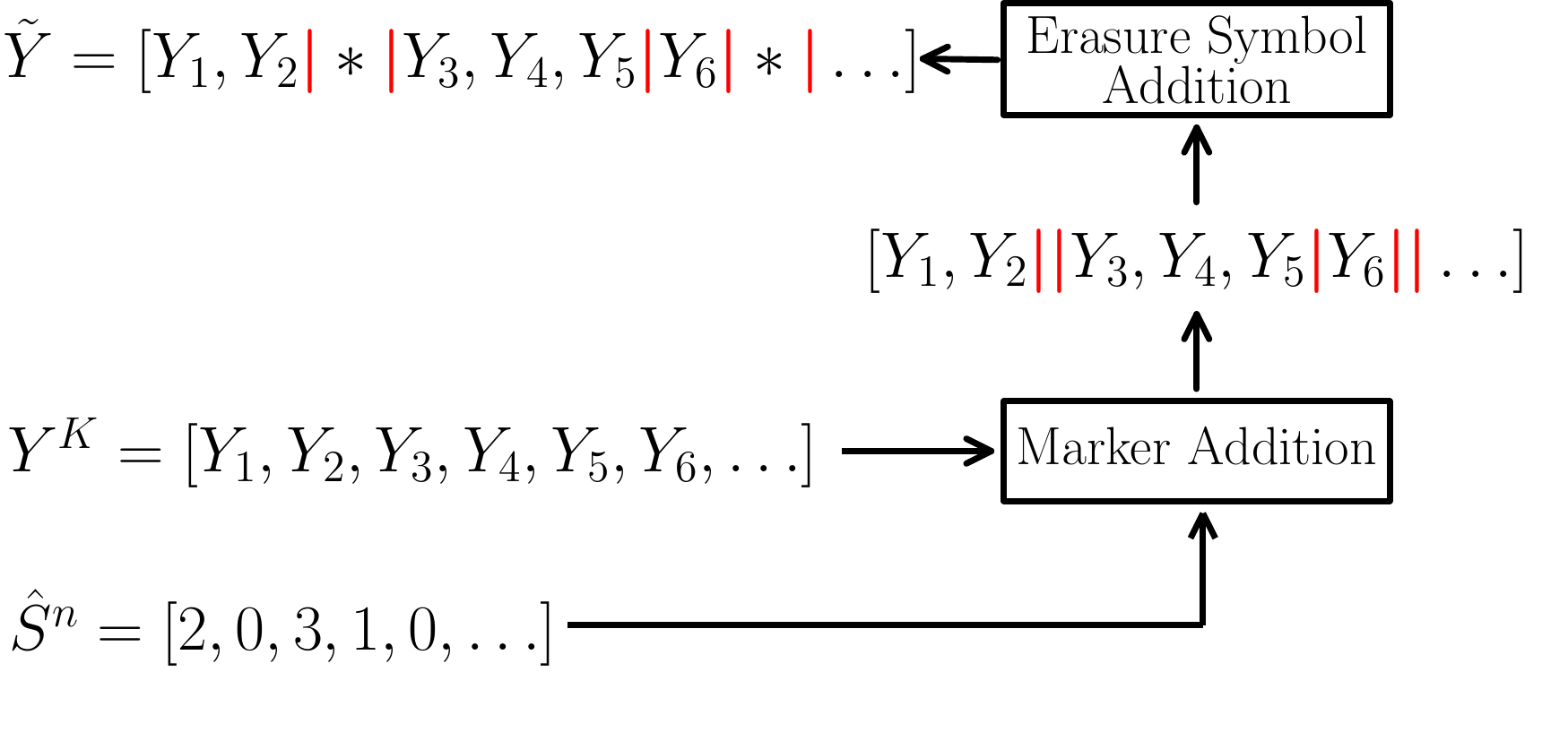}}
\caption{An example of the construction of $\tilde{\mathbf{D}}^{(2)}$, as described in Step~3 of the proof of Theorem~\ref{thm:mainresult}, illustrated over a pair of rows $X^n$ of $\mathbf{D}^{(1)}$ and $Y^K$ of $\mathbf{D}^{(2)}$. After these steps, in Step~4 we check the joint typicality of the rows $X^n$ of $\mathbf{D}^{(1)}$ and $\tilde{Y}$ of $\tilde{\mathbf{D}}^{(2)}$.}
\label{fig:marker}
\end{figure}
\item Fix $\epsilon>0$. Match the $l$\textsuperscript{th} row $Y^K_{l}$ of $\tilde{\mathbf{D}}^{(2)}$ with the $i$\textsuperscript{th} row $X^n_i$ of {${\mathbf{D}}^{(1)}$}, if $X_i$ is the only row of {${\mathbf{D}}^{(1)}$} jointly $\epsilon$-typical with $Y^K_l$ according to $p_{X,Y^S,S}$, assigning $\hat\Theta_n(i)=l$, where
{
 \begin{align}
    p_{X,{Y}^S|S}(x,y^s|s)&=\begin{cases}
      p_X(x) \mathds{1}_{[y^s = \ast]} &\text{if } s=0\\
      p_X(x) \prod\limits_{j=1}^s p_{Y|X}(y_j|x) &\text{if } s\ge 1
    \end{cases}
\end{align}
}
with $y^s=y_1\dots y_s$.
Otherwise, declare an error.
\end{enumerate}

The column discarding and the marker addition as described in Steps~3-4, are illustrated in Figure~\ref{fig:marker}.

The total probability of error of this scheme (as in~\eqref{eq:proberror}) can be bounded as follows
\begin{align}
    P_e 
    &\le 2^{n R}  2^{-n(I(X;Y^S,S)-3 \epsilon)}+\epsilon+\rho_n+\mu_n
\end{align}
Note that since $m_n$ is exponential in $n$, $d\ge1$, and from WLLN, using Lemma~\ref{lem:noisyreplicadetection} we have $\rho_n\to0$ and using Lemma~\ref{lem:seededdeletiondetection} we have $\mu_n\to0$ as $n\to\infty$. Thus $P_e\le \epsilon$ as $n\to\infty$ if $R<I(X;Y^S,S)$, concluding the proof. 
\end{sketchproofachievable}

The matching scheme proposed above for noisy repeated database matching is different from the one proposed in~\cite{noiselesslonger} for the noiseless setting in several ways: First, in the noiseless setting, the seeds are not required and a single detection algorithm can identify deletions and replicas. Second, in Step~3 of the proof above, unlike~\cite{noiselesslonger}, the noisy replicas are retained. This is because under noise, replicas offer additional information, similar to a repetition code. This implies an important distinction between database matching under synchronization errors and decoding in a repeat channel~\cite{cheraghchi2020overview}: In database matching, the identical repetition pattern over a large number of rows allows us to detect deletions and replicas, which in turn improves the achievable database growth rate. On the other hand, in a repeat channel, detecting the repetition pattern is in general not possible and the replicas have a negative impact on the channel capacity.

\subsection{Converse}
\label{subsec:converse}
We argue that the database growth rate achieved in Theorem~\ref{thm:mainresult} is in fact tight using a genie-aided proof through Fano's inequality where the repetition pattern $S^n$ is known. We argue that since the rows are \emph{i.i.d.} conditioned on the repetition pattern $S^n$, the seeds $(\mathbf{G}^{(1)},\mathbf{G}^{(2)})$ do not offer any additional information given $S^n$. Therefore, as the seeds become irrelevant in this genie-aided proof, we argue that the converse result holds for any seed order $d$.

\begin{proofconverse}
Let $R$ be the database growth rate and $P_e$ be the probability that the scheme is unsuccessful for a uniformly-selected row pair. More formally,
\begin{align}
   P_e&\triangleq \Pr\left(\boldsymbol{\Theta}_n(J)\neq\hat{\boldsymbol{\Theta}}_n(J)\right),\hspace{1em} J\sim\text{Unif}([m_n])
\end{align}
Furthermore, let $S^n$ be the repetition pattern and $K=\sum_{j=1}^n S_j$. Since $\boldsymbol{\Theta}_n$ is a uniform permutation, from Fano's inequality, we have
\vspace{-1em}
\begin{adjustwidth}{-0.2cm}{0pt}
\begin{align}
    H(\boldsymbol{\Theta})&\le 1+m_n P_e\log m_n+I(\boldsymbol{\Theta}_n;\mathbf{D}^{(1)},\mathbf{D}^{(2)},\mathbf{G}^{(1)},\mathbf{G}^{(2)})
\end{align}
\end{adjustwidth}
From the independence of $\boldsymbol{\Theta}_n$, $\mathbf{D}^{(2)}$ and $(\mathbf{G}^{(1)},\mathbf{G}^{(2)})$, we get
\vspace{-1em}
\begin{adjustwidth}{-0.3cm}{0pt}
\begin{align}
    I(\boldsymbol{\Theta}_n;\mathbf{D}^{(1)},\mathbf{D}^{(2)},\mathbf{G}^{(1)},\mathbf{G}^{(2)})
    &= I(\boldsymbol{\Theta}_n;\mathbf{D}^{(1)}|\mathbf{D}^{(2)},\mathbf{G}^{(1)},\mathbf{G}^{(2)})\\
    &\le I(\boldsymbol{\Theta}_n,\mathbf{D}^{(2)},\mathbf{G}^{(1)},\mathbf{G}^{(2)};\mathbf{D}^{(1)})\\
    &\le I(\boldsymbol{\Theta}_n,\mathbf{D}^{(2)},S^n;\mathbf{D}^{(1)})\label{eq:converseassumeS}\\
    &= m_n  I(Y^K,S^n;X^n)\label{eq:iidrows}\\
    &= m_n n I(X;Y^S,S)\label{eq:iidentries}
\end{align}
\end{adjustwidth}
where \eqref{eq:converseassumeS} follows from the fact that given $S^n$, $\mathbf{G}^{(1)},\mathbf{G}^{(2)}$ do not offer any additional information. Equation~\eqref{eq:iidrows} follows from the fact that non-matching rows are \emph{i.i.d.} conditioned on the repetition pattern ${S}^n$. Furthermore, \eqref{eq:iidentries} follows from the fact that the entries of $\mathbf{D}^{(1)}$ \emph{i.i.d.}, and the noise on the entries are also \emph{i.i.d}. 

Finally, from Stirling's approximation and \eqref{eq:iidentries}, we obtain
\begin{align}
    R&= \lim\limits_{n\to\infty}\frac{1}{m_n n}H(\boldsymbol{\Theta_n})\\
    &\le \lim\limits_{n\to\infty}\left[ \frac{1}{m_n n}+P_e R+I(X;Y^S,S)\right]\\
    &\le I(X;Y^S,S) \label{eqn:converselast}
\end{align}
where \eqref{eqn:converselast} follows from the fact that $P_e\to 0$ as $n\to\infty$.
\end{proofconverse}

\section{Conclusion}
\label{sec:conclusion}
In this work, we have studied the database matching problem under random noisy column repetitions. We have showed that the running Hamming distances between the consecutive columns of the labeled noisy repeated database can be used to detect replicas. In addition, given seeds whose size grows logarithmic with the number of rows, an exhaustive search over the deletion patterns can be used to infer the locations of the deletions. Using the proposed detection algorithms, and a joint typicality based rowwise matching scheme, we have derived an achievable database growth rate, which we prove is tight. Therefore, we have completely characterized the database matching capacity under noisy column repetitions. 

\typeout{}
\bibliographystyle{IEEEtran}
\bibliography{references}
\end{document}